\documentclass[11pt]{article}

\usepackage{latexsym}
\usepackage{amsmath}
\usepackage{amsfonts}
\usepackage{amssymb}
\usepackage{amsthm}
\usepackage{graphicx}
\usepackage{multirow}
\usepackage{rotating}
\usepackage{pbox}
\usepackage[defaultlines=2, all]{nowidow}

\newtheorem{theorem}{Theorem}
\newtheorem{lemma}[theorem]{Lemma}
\newtheorem{proposition}[theorem]{Proposition}

\theoremstyle{definition}
\theoremstyle{problem}
\newtheorem{definition}{Definition}
\newtheorem{example}{Example}
\newtheorem{problem}{Problem}

\newcommand{\payoff}{{{\phi}}}
\newcommand{\bid}{{{b}}}
\newcommand{\coal}{{{\mathcal{C}}}}
\newcommand{\tot}{{{\mathit{tot}}}}
\newcommand{\naturals}{{{\mathbb{N}}}}

\newcommand{\fpt}{{\mathrm{FPT}}}
\newcommand{\wone}{{\mathrm{W[1]}}}

\newcommand{\cost}{{{c}}}
\newcommand{\salary}{{{\phi}}}
\newcommand{\minsalary}{{{\phi_{i}^{\min}}}}
\newcommand{\iminsalary}[1]{{{\phi_{#1}^{\min}}}}

\begin{document}

\title{Cooperation and Competition when Bidding for Complex Projects: Centralized and Decentralized Perspectives}
\date{}

\author{Piotr Skowron \\
   University of Oxford \\
   Email: piotr.skowron@cs.ox.ac.uk \\
\and
   Krzysztof Rzadca \\
   University of Warsaw \\
   Email: krzadca@mimuw.edu.pl
\and
   Anwitaman Datta \\
   Nanyang Technological University \\
   Email: anwitaman@ntu.edu.sg
}

\maketitle

\begin{abstract}
To successfully complete a complex project, be it a construction of an airport or of a backbone IT system, agents (companies or individuals) must form a team having required competences and resources. A team can be formed either by the project issuer based on individual agents' offers (centralized formation); or by the agents themselves (decentralized formation) bidding for a project as a consortium---in that case many feasible teams compete for the contract. We investigate rational strategies of the agents (what salary should they ask? with whom should they team up?). We propose concepts to characterize the stability of the winning teams and study their computational complexity.
\end{abstract}

\section{Introduction}

Consider a complex project: involved, intricate, and consisting of many varied yet interrelated parts. The successful completion of such a project requires coordinated cooperation of a number of experts---people and companies---often organized as teams of subcontractors~\cite{Williams1999269}. For instance, in the construction industry, to build an apartment building (a rather standard endeavor), typically, 30 to 40 individual sub-contractors are involved in 100 to 150 separate activities~\cite{walshConstuction}.

Assigning sub-tasks of a complex project to subcontractors is common. 
In the UK, the proportions of construction employees employed by sub-contractors in years 1983--1998 has grown by 20\%~\cite{doi:10.1142/S1609945103000406}. In the UK, between 2008 and 2011, the number of freelancers increased by 12\%; in Australia in 2012, 17.2\% of the workforce were self-employed (8.5\% as independent contractors). These are only a few examples of a growing tendency to develop  projects by employing many specialized sub-contractors instead of a single company.

Nevertheless, it is not clear how to organize the market both for the issuer of the project (in this paper, called the client) and the subcontractors (later called the agents).  Interaction between the agents applying for the employment in a project and the client is captured by the \emph{hiring a team} problem~\cite{Chen:2007:CLE:1283383.1283459, Iwasaki:2007:FMH:1781894.1781924, Chen:2010:FMD:1917827.1918350}. The agents have private costs of participating in the project and may have different skills, thus only certain teams are able to complete the project on time. The client organizes an auction in which individual agents place their bids, i.e., their required salaries. After collecting the bids, the client selects the cheapest feasible team, i.e., the set of agents able to complete the project on time with the lowest total bid.

We generalize the hiring a team problem by exploring two organizations of the market. The original approach corresponds to the \emph{centralized setting}: the agents communicate only with the client by issuing their required salaries (also referred to as bids) and it is the client's responsibility to select an appropriate team.
Our main contribution lies in considering a different organization of the market, where \emph{the agents first form teams and then bid for the project as consolidated groups} rather than as individuals. Since the organization of the agents into teams is not managed by a central entity, we refer to this setting as \emph{decentralized}. To the best of our knowledge this formulation of the problem is novel and leads to a new class of games. This formulation has a natural interpretation: a client may not want to coordinate a project and to deal with individual subcontractors, but instead expects that the subcontractors coordinate among themselves and propose a bid for completing the whole project.

Additionally, we generalize the hiring a team problem by considering two types of agents' compensation. In the \emph{project salary model} (corresponding to the original approach) the agents are payed for their work irrespectively of the contributed effort. We propose a new payment model, the \emph{hourly salary}, in which the agents are payed for the time spent working on the project.


Throughout this paper we assume that we are given an oracle that, for a given team, can determine whether this team is feasible, i.e., whether it can successfully complete the project. In particular, given a budget and the requested agents' salaries, the oracle can be asked to find a feasible team or to find the cheapest feasible team. In Appendix~\ref{sec:findingFeasible} we show how to create such an oracle for a concrete scheduling model (which moreover generalizes a commodity auction); we also determine its exact complexity.

Our approach generalizes two models: commodity auctions~\cite{mcCabeThesis} and path auctions~\cite{Nisan:1999:AMD:301250.301287} (Appendix~\ref{sec:findingFeasible}).
In a commodity auction, there is a set of items $I = \{i_1, i_2, \dots, i_q\}$ and agents owning certain subsets of $I$. A team is feasible if the agents have together all the items from $I$. A commodity auction can be mapped to our problem by considering that $I$ is a set of independent activities; an agent owning a subset corresponds to an agent having skills to complete these activities.
In a path auction, there is a graph $G$ with two distinguished vertices: a source $s$ and a target $t$. The agents correspond to the vertices in the graph; some vertices are connected by edges. A team is feasible if the participating agents form a path from $s$ to $t$.

Since we consider teams of agents with sufficient skills, our model resembles cooperative skill games~\cite{Bachrach20131} and coalitional resource games~\cite{Wooldridge:2006:CCC:1148349.1148351}. These games, however, consider the stability of the \emph{grand coalition} and interaction between its members. In contrast, our approach is to \emph{expose the competition between multiple teams}. Thus, we do not apply the typical cooperative game theory concepts and, instead, model agents' cooperation and competition as a non-cooperative game (see Appendix~\ref{sec:otherSolConcepts} for a detailed comparison). Our approach is thus closer to endogenous formation of coalitions~\cite{hart1983endogenous, bloch1996sequential, ray1999theory}.

In this paper we identify and formalize a new class of coalition games, which are the extensions of the hiring a team games. We propose concepts to characterize the stability of the winning teams and study their computational complexity. Table~\ref{tab:summary} summarizes our results. All the proofs omitted from the main text are provided in the appendix.

\begin{table}[t]
\small
\caption{Summary of our results. ``Existence'' denotes whether a team/equilibrium always exists. ``Checking'' gives the complexity of checking whether a given team satisfies the definition. ``Finding''  gives the complexity of finding a team/equilibrium. $\textit{FFT}$ and $\textit{FCFT}$ are the complexities of the problems \textsc{FFT} and \textsc{FCFT}, respectively. The symbols $\dagger$ (respectively, $\mathsection$) denotes that a result is valid only in the project (respectively, hourly) salary model. The symbol $\Diamond$ denotes that a result is valid only if the salaries of the agents are rational numbers. Whenever one of the symbols $\dagger, \mathsection$ or $\Diamond$ is provided, it means that the problem for the other cases is still open.}
\label{tab:summary}
\begin{center}
\begin{tabular}{p{0.1cm}p{5.5cm}|c|c|c}
\multicolumn{2}{c|}{solution concept} & Existence & Checking & Finding \\ \cline{1-5}
\multicolumn{1}{ c }{\begin{rotate}{270}\hspace{-0.7em}decentral.\end{rotate}}
& rigor. strongly winning team & Not always & $O(n^2 \cdot \textit{FCFT})$ & $O(n^5\log(nv)\textit{FCFT})$ $\dagger$ $\Diamond$\\ \cline{2-5}
& strongly winning team & Not always & \multicolumn{2}{ |c }{open problem}\\ \cline{2-5}
& weakly winning team & Always & \multicolumn{2}{ |c }{$O(n^5\log(nv)\textit{FCFT})$ $\dagger$ $\Diamond$}\\ \cline{2-5}
& auction winning team & Always & $O(\textit{FFT})$ & $O(v \cdot \textit{FFT})$ \\ \cline{1-5}
\multicolumn{1}{ c }{\begin{rotate}{270}\hspace{-0.7em}central.\end{rotate}}
& winning team (with asking salaries) & N/A & \multicolumn{2}{ |c }{$O(\textit{FCFT})$} \\ \cline{2-5}
& Strong Nash Equilibrium & \pbox{20cm} {Always $\dagger$ \\ Not always $\mathsection$} & $O(\textit{FCFT})$ & $O(n^3\log(nv)\textit{FCFT}))$ $\dagger$ $\Diamond$\\ 
\end{tabular}
\end{center}
\end{table}

\section{A Complex Project as a Game: a Formal Model}\label{sec:model}

We consider a game in which a client (an issuer) submits a single complex project.
The client has a certain valuation $v$ of the project that is the maximal price that she is able to pay for completing the project.

There is a set $N = \{1, 2, \dots n\}$ of $n$ agents.
For each agent $i$, we define $\minsalary > 0$ to be the agent's \emph{minimal salary} for which $i$ is willing to work. This minimal salary may correspond to the agent's personal cost of participating in the project. The agent prefers to work for $\minsalary$ than not to work (and then to work for higher salary).
The value $\minsalary$ is private to the agent---neither the issuer nor the other agents know $\minsalary$.

A subset of the agents' population $N$ forms a team to work on the project; the paper's core contribution is on how this process should be organized. 
A team $\coal$ is a triple $\langle N_{\coal}, \salary_\coal, \cost_\coal \rangle$ consisting of: the set of participating agents $N_{\coal} \subseteq N$;  
a salary function $\salary_{\coal}: N_{\coal} \rightarrow \naturals$ assigning salaries to member agents; 
and the total cost of the team $\cost_{\coal} \in \naturals$---the total amount of money earned by the participants of $\coal$. Salaries are discrete (not only money is discrete, but also it is common in real-world auctions to specify a minimal difference between two successive bids). However, to derive some computational results, in some clearly marked places, we assume that the salaries can be rational numbers.

The same team may organize the work of its members on the project in various ways with varying efforts from participants. To capture this property, we introduce a notion of a \emph{schedule}, $\sigma_{\coal}: N_{\coal} \rightarrow \naturals$, that assigns to each member of a team the amount of time this agent needs to spend on the project. Of course, there may exist many  schedules for a single team. We expand the discussion on the notion of schedule in Appendix~\ref{sec:findingFeasible}.

We consider two models of agents' compensation. 
Let $\salary_{\coal}^{\tot}(i)$ denote the total amount of money agent $i$ gets in team $\coal$  (naturally, $\cost_{\coal} = \sum_{i \in N_{\coal}} \salary_{\coal}^{\tot}(i)$). 
In the \emph{project salary} model $\salary_{\coal}^{\tot}(i)$ is equal to the salary of the agent $\salary_{\coal}(i)$ (and thus does not depend on the amount of work assigned to that agent).
In the \emph{hourly salary} model  $\salary_{\coal}^{\tot}(i)$ is equal to the product of the salary $\salary_{\coal}(i)$ and the time $t_i$ during which $i$ processes her part of the project ($t_i$ is known from the schedule).

In the project salary model the agents are interested in earning as much money as possible. 
The hourly salary model represents agents who are interested in having the highest possible hourly wage; thus, e.g., an agent prefers to work $t_i = 1$ time unit with a salary $\salary_i=3$ to working $t_i = 2$ time units with a salary $\salary_i=2$.

Different schedules might result in different completion times of the project. If the schedule results in a completion time that is satisfactory for the client, we say that the schedule is \emph{feasible}. For some teams there might not exist a feasible schedule (e.g., if the members lack certain skills). We assume that there is an oracle that can answer whether a given schedule is feasible or not. This very general setting can be instantiated by providing a concrete oracle. For instance, in Appendix~\ref{sec:findingFeasible} we show that by appropriately specifying the oracle, our results can be applied to commodity auctions and to path auctions. We also show there how to replace the general oracle with a concrete scheduling model.

A team $\coal$ is \emph{feasible} iff (i) the asking salaries are no-lower than the minimal salaries, $\salary_{\coal}(i) \geq \minsalary$; and there exist a feasible schedule such that: (ii) the project budget is not exceeded ($\cost_\coal \leq v$), and (iii) the cost $\cost_{\coal}$ of the team $\coal$ is consistent with the salaries $\salary_\coal$. 
Specifically, in the project salary model $\cost_{\coal} = \sum_{i \in N_{\coal}} \salary_{\coal}(i)$. 
In the hourly salary model $\cost_{\coal} = \sum_{i \in N_{\coal}} t_i \salary_{\coal}(i)$.

A team $\coal$ is \emph{cheaper} than $\coal'$ if it has a strictly lower cost  $\cost_{\coal} < \cost_{\coal'}$ or if it has the same cost, but it is preferred by a deterministic tie-breaking rule $\prec$, $N_{\coal} \prec N_{\coal'}$ (for the sake of concreteness we assume that $\prec$ is the lexicographic order in which a team is represented by a concatenation of the sorted list of the names of its members).

Throughout this paper we use the \textsc{Find Feasible Team (FFT)} and \textsc{Find Cheapest Feasible Team (FCFT)} problems.
\begin{problem}
An instance of  Find Feasible Team (FFT) consists of a project (with a budget $v$) and the set $N$ of the agents with (known) minimal required salaries $\minsalary$.
The question is to find some feasible team or to claim there is no such. In the Find Cheapest Feasible Team (FCFT) we ask for the cheapest feasible team.
\end{problem}

\section{Centralized Formation of Teams}\label{sec:nonStrategic}

In the centralized model agents submit their asking salaries $\salary_i$ directly to the client. The client, having the asking salaries, wants to form the cheapest feasible team. We first show that this problem reduces to \textsc{FFT}, the problem of finding a feasible team. Then, we analyze the optimal bidding strategies of agents.

\begin{proposition}\label{prop:centralWinningTeamRed}
The problem \textsc{FCFT} can be solved in time $O((\log v + n)\textsc{FFT})$, where $\textsc{FFT}$ is the complexity of the problem \textsc{FFT}. Having the asking salaries of the agents, the problem of finding the winning team can be solved in time $O(\textit{FCFT})$, where $\textit{FCFT}$ is the complexity of the problem \textsc{FCFT}.
\end{proposition}

The agents may behave strategically and manipulate their asking salaries to maximize their payoffs. 
We model this problem as a strategic game. An action of agent $i$ is her asking salary $\salary_i \geq \minsalary$. The payoff of $i$ is $\salary_i$ iff $i$ is a member of the cheapest feasible team; otherwise the payoff of $i$ is 0. 

Interestingly, in the project salary model, there exist sets of vectors of actions which are stable against collaborative strategies of the agents. We recall that a vector of the agents' actions is a Strong Nash Equilibrium (SNE) if no subset of the agents can change its actions so that all the deviating agents obtain strictly better payoffs. 

For each subset of the agents $N' \subseteq N$, by $\coal^{*}(N')$ we denote the cheapest feasible team using only the agents from $N'$ 
(if there is no feasible team consisting of the agents from $N'$, the team $\coal^{*}(N')$ does not exist).

\begin{theorem}\label{thm:existenceOfSNE}
In the project salary model, if there exists a feasible team then there exists a Strong Nash Equilibrium. In every SNE, the set of the agents who get positive payoffs is the set of agents forming the cheapest feasible team, $N_{\coal^{*}(N)}$.
\end{theorem}
\begin{proof}
Let $N^{*} = N_{\coal^{*}(N)}$ be the set of the agents participating in the cheapest feasible team. We say that the action $\salary_i$ of the agent $i$ is \emph{minimal} if and only if $\salary_i = \iminsalary{i}$.
We show how to construct the asking salaries $\salary_i^{*}$ of the agents from $N^{*}$ that, together with the minimal actions of the agents outside $N^{*}$, form a Strong Nash Equilibrium. 
A sketch of the proof is as follows. We show the set of linear inequalities for the variables $\salary_i, i \in N^{*}$. Let us denote the maximal values of $\salary_i$ which satisfy the inequalities as $\salary_i^{*}$ (maximal in the sense that if we increase any value $\salary_i^{*}$, then the new values will not satisfy all the inequalities any more). We show that the actions $\salary_i^{*}$ of the agents from $N^{*}$, together with the minimal actions of the agents outside of $N^{*}$, form an SNE and that the set of the solutions $\salary_i^{*}$ that satisfy all the inequalities is nonempty.

The first inequality states that the values $\salary_i$ must lead to a feasible solution:
\begin{align}
\sum_{i \in N^{*}} \salary_i \leq v \textrm{.}
\label{inq:salary}\end{align}
Next, as $\coal^*$ is the cheapest feasible team, for each feasible team $\coal'$ ($N^{*} \neq N_{\coal'}$) such that $N^{*} \prec N_{\coal'}$, $\coal^*$ must have (weakly) lower cost:
\begin{align}
\sum_{i \in N^{*} \setminus N_{\coal'}} \salary_i \leq \sum_{i \in N_{\coal'} \setminus N^{*}} \minsalary \textrm{.}
\label{inq:cheap1}\end{align}
For a $\coal'$ preferred over $\coal^*$ ($N^{*} \neq N_{\coal'}$ and $N_{\coal'} \prec N^{*}$), $\coal^*$ must have strongly lower cost:
\begin{align}
\sum_{i \in N^{*} \setminus N_{\coal'}} \salary_i < \sum_{i \in N_{\coal'} \setminus N^{*}} \minsalary \textrm{.}
\label{inq:cheap2}\end{align}

First, if the values $\salary_i^{*}$ satisfy inequalities (\ref{inq:salary})-(\ref{inq:cheap2}) and the agents outside of $N^{*}$ play their minimal actions, then the agents from $N^{*}$ will get positive payoffs.
If they did not get the positive payoffs, it would mean that there exists a feasible cheaper team $\coal'$. However, inequalities (\ref{inq:cheap1})-(\ref{inq:cheap2}) imply that the agents from $N^{*} \setminus N_{\coal'}$ induce the lower total cost than the total cost of the agents from $N_{\coal'} \setminus N^{*}$; this ensures that agents $N^{*}$ with actions $\salary_i^{*}$ form a cheaper team than $\coal'$.

Next, we show that no set of agents $N_{\coal'}$ can make a collaborative action $\overline{\salary}$, after which the payoff for all $N_{\coal'}$ agents will be greater than previously. By contradiction, assume that there exists such a set of agents $N_{\coal'}$ and such an action $\overline{\salary}$. First we consider the case when the payoff of some agent $i \notin N^{*}$ would change. This means that after $\overline{\salary}$ there would be a new cheapest feasible team $\coal'$, where $i \in N_{\coal'}$. However, we know that the total cost of the agents from $N^{*} \setminus N_{\coal'}$ is lower than the total cost of the agents from $N_{\coal'} \setminus N^{*}$. This means that $\coal'$ cannot be cheaper than the team consisting of the agents from $N^{*}$. Finally, consider the case when only the payoffs of the agents from $N^{*}$ change (and thus $N_{\coal'} \subseteq N^{*}$). 
However, if the strict subset of $N^{*}$ could form a feasible team, then $\coal^{*}(N)$ would not be the cheapest. Thus, $N_{\coal'} = N^{*}$.
This means that every agent from $N^{*}$ must have played a higher action (and others must have not changed their actions). Since $\salary_i^{*}$ were maximal, this means that after the action $\overline{\salary}$ some inequality, for some feasible team $\coal''$, would not hold any more. Thus, we infer that $\coal''$ is cheaper than $\coal'$.

To check that there always exists a solution, we see that the definition of $N^{*}$ ensures that the values $\salary_i^{*} = \minsalary$ satisfy all inequalities. 

Finally, by contradiction we prove the $N^*$ is formed by the same agents as forming the cheapest team. Assume that the set of the agents that get positive payoffs in some SNE is $N' \neq N^{*}$. However, if the agents from $(N^{*} \setminus N')$ play their minimal actions, then the team consisting of the agents from $N^{*}$ would be cheaper than the team consisting of the agents from $N'$. Thus, the agents from $(N^{*} \setminus N')$ can deviate, getting better payoffs. This completes the proof.
\end{proof}

Interestingly, there is no analogous result for hourly salary model (see Proposition~\ref{prop:sneNotExistsInHourly} in Appendix~
\ref{app:proofs}). 
The proof of Theorem~\ref{thm:existenceOfSNE} is constructive, but it requires considering all feasible teams and, so, leads to potentially high computational complexity. Finding an efficient algorithm for the problem of finding Strong Nash Equilibria in the project salary model is open.
On the other hand, if the salaries of the agents can be rational numbers, we can find the salary function in SNE by a polynomial reduction to the \textsc{FCFT} problem. This result is particularly meaningful if the salaries have high granularity; rounding such a rational solution gives an integral solution which is nearly perfect.

\begin{proposition}\label{proptwo}
In the project salary model, if the salaries are rational, then finding a Strong Nash Equilibrium can be solved in time $O(n^3\log(nv)\textit{FCFT}))$, where $\textit{FCFT}$ is the complexity of the problem \textsc{FCFT}. Checking whether a given vector of the asking salaries $\langle \salary_i \rangle, i \in N$ is a Strong Nash Equilibrium can be solved in time $O(\textit{FCFT})$, where $\textit{FCFT}$ is the complexity of the problem \textsc{FCFT}.
\end{proposition}

\section{Decentralized Formation of Teams}\label{sec:winningCoalitions}

If the agents can communicate and coordinate their strategies, they form teams and bid for the project as consortiums. We propose the concept of a (rigorously) strongly winning team, in which no subset of agents can successfully deviate. We show how to characterize (rigorously) strongly winning teams and how to reduce the problem of finding them to the \textsc{FCFT} problem. We show that the strongly winning teams may not exist, and so we introduce the concept of a weakly winning team. We prove that a weakly winning team always exists (provided that there is a feasible team). We demonstrate how to reduce the problem of finding weakly winning teams to the \textsc{FCFT} problem.

We model the behavior of the agents as a strategic game. Agent $i$'s action is a triple $\langle N_{\coal}, \salary_{\coal}, \bid_\coal \rangle$. Intuitively, such an action means that the agent $i$ decides to enter the team $\coal = \langle N_{\coal}, \salary_{\coal}, \bid_\coal \rangle$. The payoff of the agent is equal to $\salary_{\coal}(i)$ if (i) $\coal$ is feasible, (ii) each agent $j \in N_{\coal}$ agrees to participate in $\coal$ (i.e., they all play $\coal$, and their payoffs are consistent with the bid of the team $\bid_\coal$), and (iii) there is no feasible cheaper team $\coal'$ such that all the agents from $N_{\coal'}$ agree to participate in $\coal'$. Otherwise, the payoff of $i$ is 0.

\subsection{Strongly Winning Teams}

As the payoffs depend on whether the others agree to cooperate, 
rather than the Nash Equilibrium, the Strong Nash Equilibrium (SNE) should be used.
In the following definition we propose an even more stable equilibrium concept---the Rigorously Strong Nash Equilibrium (RSNE), which requires that no subset of agents can deviate such that each agent gets a payoff \emph{at least as good} as its payoff before deviating (instead of SNE's strictly better). Our approach is motivated by cautious agents. In an SNE, the agents have no incentive to deviate if they get the same payoff; however they also have no incentive not to deviate. Yet, any deviation will result in a serious payoff loss for some agents (changing their payoffs from a positive $\salary$ to zero). A cautious agent will prefer not to be exposed to the possibility of such a loss.

\begin{definition}
The vector of actions $\pi$ is a \emph{Rigorously Strong Nash Equilibrium (RSNE)} iff there is no subset of agents $N_{\coal}$ such that the agents from $N_{\coal}$ can make a collaborative action $\coal$ after which the payoff of each agent $i$ from $N_{\coal}$ would be at least equal to her payoff under $\pi$ and the payoff of at least one agent $i \in N$ would improve.
\end{definition}

A RSNE requires that the payoff of at least one agent $i \in N$ must change as we treat as equivalent the teams with the same payoffs. For instance, in a game with three agents, $a$, $b$ and $c$, if the team $\{a,b\}$ gets a positive payoff, it does not matter whether $c$ plays $\langle \{c\}, v+1 \rangle $ or $\langle \emptyset, v+1 \rangle$: in both cases all payoffs are the same (recall that $v$ is the client's maximal budget for the project).

Below we introduce additional definitions that help characterize the RSNE in our games.

\begin{definition}\label{def:endangered}
A feasible team $\coal$ is \emph{explicitly endangered} by a team $\coal'$ if (i) $\coal'$ is feasible, (ii) $N_{\coal} \cap N_{\coal'} = \emptyset$ and (iii) $\coal'$ is cheaper than $\coal$.
A feasible team $\coal$ is \emph{implicitly endangered} by a team $\coal'$ if (i) $\coal'$ is feasible, (ii) $N_{\coal} \cap N_{\coal'} \neq \emptyset$ and each agent from $N_{\coal} \cap N_{\coal'}$ gets in $\coal'$ at least as good a salary 
as in $\coal$, and (iii) either $N_{\coal} \neq N_{\coal'}$ or $\salary_{\coal} \neq \salary_{\coal'}$.
\end{definition}

If there are agents belonging to both teams ($N_{\coal} \cap N_{\coal'} \neq \emptyset$), we do not consider the total cost of the alternative team $\coal'$, as the decision whether $\coal'$ will be formed depends solely on the agents from $N_{\coal} \cap N_{\coal'}$: if they decide to form $\coal'$, $\coal$ will not be formed, thus the client won't be able to choose between $\coal$ and $\coal'$.

A feasible team $\coal$ is \emph{(rigorously) strongly winning} iff there is a (Rigorously) Strong Nash Equilibrium in which the agents from $N_{\coal}$ get positive payoffs $\salary_{\coal}$.
The following theorem relates endangerment (Definition~\ref{def:endangered}) and a winning team.
\begin{theorem}\label{thm:rigorWinningCoal}
The team $\coal$ is rigorously strongly winning if and only if $\coal$ is not explicitly nor implicitly endangered by any team.
\end{theorem}

The result in Theorem~\ref{thm:winningCoal} stated for RSNEs transfers to SNEs after a slight modification of the payoffs. It is sufficient to assume that an agent playing an empty team receives slightly higher payoff than if she plays a non-empty losing team. In other words, this modification associates some small costs with the preparation of a bid by the agents. Hereinafter, whenever we mention a strictly winning team we assume that the agents incur such costs. To state the result for SNEs we also need to use the definition of a team $\coal$ being \emph{strictly implicitly endangered} by $\coal'$. This definition differs from being implicitly endangered only by not requiring the agents from $N_{\coal} \cap N_{\coal'}$ to have at least as good payoffs, but strictly better payoffs in $\coal'$ than in $\coal$.

\begin{theorem}\label{thm:winningCoal}
If there are small but positive costs of preparing the offer by the agents then
the team $\coal$ is strongly winning if and only if $\coal$ is not explicitly nor strictly implicitly endangered by any team.
\end{theorem}

Theorems~\ref{thm:rigorWinningCoal}~and~\ref{thm:winningCoal} lead to a simple brute-force algorithm for checking whether the team $\coal$ can be a part of some RSNE. It is sufficient to check whether for each set of agents $N' \subseteq N_{\coal}$ there exists a payoff function $\salary_\coal$ and a cost $\cost$ such that $\coal$ is explicitly or implicitly endangered by $\langle N', \salary, \cost \rangle$ (such a condition can be checked by enumerating the payoff functions which assign to each agent his or her minimal salary, the salary that he or she obtains in $\coal$, or the next higher salary).  
Below, we characterize RSNEs in the project salary model even more precisely.

\begin{lemma}\label{thm:rigorStrongWinningCharacter}
In the project salary model, the set of agents participating in a rigorously strongly winning team is the same as the set of agents participating in the cheapest feasible team.
\end{lemma}
\begin{lemma}\label{prop:maxBid}
In the project salary model
the bid of a strongly winning team is equal to the maximal allowed price $v$.
\end{lemma}

Lemma~\ref{thm:rigorStrongWinningCharacter}~and~Lemma~\ref{prop:maxBid} show that the problem of finding a strongly winning team reduces to the problem of finding a feasible team. The problem, thus, becomes an optimization problem; the strategic behavior of agents has no impact (see Propositions~\ref{prop:checkingRigStrWinn}~and~\ref{prop:findingRigStrWinn} in Appendix~\ref{app:proofs}). An RSNE (and even an SNE) may not exist in some instances.

\begin{proposition}\label{prop:noWinningCoal}
Both in the project salary and in the hourly salary model, there may not exist a strongly winning team even though there exists a feasible team.
\end{proposition}
\begin{proof}
Consider a project with budget $v=5$; and three identical agents $a$, $b$, $c$ with minimal salaries $\minsalary=2$ (in the hourly salary model, assume that each agent spends
exactly 1 time unit on the project); a team of any two agents is feasible (able to complete the project on time and within the budget).

For the sake of contradiction assume there exists a team $\coal$ that gets positive payoffs. Without loss of generality we assume that $N_{\coal} = \{ a, b \}$. At least one of the agents, let us say $a$, has to get salary at most equal to $2.5$. However, the agents $a$ and $c$, with the salaries equal to $3$ and $2$ respectively, can form a feasible team in which both $a$ and $c$ get better payoffs.
\end{proof}

\subsection{Weakly Winning Teams}

Proposition~\ref{prop:noWinningCoal} suggests that a notion of a strongly winning team is too restrictive. The team $\{a, c\}$ can profit by deviating, e.g., by playing $\salary(a) = 3$ and $\salary(c) = 2$. But $a$ should not be willing to deviate, as $\{a, c\}$ with payoffs $\salary(a) = 3$ and $\salary(c) = 2$ too is not stable (for instance, the team $\{b, c\}$ can play $\salary(b) = 2$ and $\salary(c) = 3$, and successfully deviate from $\{a, c\}$). In the above example no team strongly wins, even though intuitively there are teams that would agree to work. Thus, we propose a weaker notion of a winning team.

\begin{definition}
A feasible team $\coal$ is \emph{weakly winning} if it is not explicitly endangered by any team and for each feasible team $\coal'$ such that $\coal$ is implicitly endangered by $\coal'$, there exists a feasible team $\coal''$ such that $\coal'$ is explicitly or implicitly endangered by $\coal''$.
\end{definition}

\begin{proposition}\label{prop:prop5stub}
There exists a weakly winning team if and only if there exists a feasible team.
\end{proposition}

\begin{proposition}\label{prop:prop6stub}
In the project salary model, if the salaries of the agents can be rational numbers, the problem of finding a weakly winning team and the problem of checking whether a team $\coal'$ is weakly winning can be solved in time $O(n^5\log(nv)\textit{FCFT})$.
\end{proposition}

Finding an efficient algorithm for the same problem with discrete salaries is still an open question.

\section{Mechanism Design}\label{sec:mechDesign}

In this section we analyze two mechanisms that a client can use to find a winning team: the first one sets the project's budget $v$; the second one uses a first-price auction.

First, we show that if the client is allowed to change the budget $v$ there exists a simple mechanism (based on a binary search) ensuring the existence of a strongly winning team.

\begin{theorem}\label{thm:newValue}
If there exists a feasible team, then there exists a budget $v^*$ for which there exists a strongly winning team.
The problem of finding such a $v^*$ can be solved in time $O(\log v \cdot \textit{FFT})$.
\end{theorem}

In the second approach we use the first-price auction in which teams participate. 
In a standard first-price auction, an item's price starts from some minimal value (the least preferred outcome for the owner of the item). Bidders place bids for the current price. The asking price is gradually increased until there are no further bids; the last bidder wins. Similarly, in our proposed auction, the auction starts from the original budget $v$ (the least preferred outcome for the client); the asking price is gradually \emph{decreased}. Teams place bids for the current asking price (as in the standard first-price auction, multiple bids for the same asking price are not allowed). The auction stops when no feasible team bids lower than the current asking price. This procedure leads to the concept of an auction-winning team.

\begin{definition}
A team $\coal$ is \emph{auction-winning} iff there is no feasible team $\coal'$ such that $b_{\coal'} < b_{\coal}$ and for each agent $i \in N_{\coal} \cap N_{\coal'}$ the agent gets better salary in $\coal'$, $\salary_{\coal'}(i) \geq \salary_{\coal}(i)$.
\end{definition}

\begin{proposition}\label{prop:findAuctionWinning}
The problem of checking whether a feasible team $\coal$ is auction-winning can be solved in time $O(\textit{FFT})$. The problem of finding an auction-winning team can be solved in time $O(v \cdot \textit{FFT})$.
\end{proposition}

\section{Conclusions}

We presented a new class of coalitional games that model cooperation and competition for employment in a complex project. Our games extend and relate to a number of well-known problems, such as coalition formation, coalitional auctions, auctions for sharable items, etc.
We considered two market organizations. In a centralized market, the winning team is selected by the client based on bids from individual agents; the agents are strategic about the salaries they request. In a decentralized market, the already-formed teams bid for the project, thus the agents are strategic both regarding their salaries and regarding their cooperation partners.

We proposed concepts of stability for each of our models and we showed how to reduce the problem of finding a winning team to the problem of finding a feasible one, for which we assumed we have an oracle with known complexity.

To instantiate our abstract model, in Appendix~\ref{sec:findingFeasible} we show how to solve a scheduling problem in which the project is a set of independent tasks and the agents have certain skills in processing them (represented by unrelated processing speeds).

\noindent\textbf{Acknowledgements:}  The authors thank Marcin Dziubi\'nski and Piotr Faliszewski for their helpful comments and Edith Elkind for the fruitful discussion on the related literature.

This research has been partly supported by the Polish National Science Center grants Sonata (UMO-2012/07/D/ST6/02440) and Preludium (UMO-2013/09/ N/ST6/03661), and by Europe Research Grant ERC-StG 639945. 

\bibliographystyle{plain}
\bibliography{crowdsourcing}

\newpage
\appendix

\section{Proofs Omitted from the Main Text}\label{app:proofs}

\theoremstyle{plain}
\newtheorem*{propzerostub}{Proposition~\ref{prop:centralWinningTeamRed}}

\begin{propzerostub}
The problem \textsc{FCFT} can be solved in time $O((\log v + n)\textsc{FFT})$, where $\textsc{FFT}$ is the complexity of the problem \textsc{FFT}. Having the asking salaries of the agents, the problem of finding the winning team can be solved in time $O(\textit{FCFT})$, where $\textit{FCFT}$ is the complexity of the problem \textsc{FCFT}.
\end{propzerostub}
\begin{proof}
We start from showing that the problem \textsc{FCFT} can be solved in time $O((\log v + n)\textsc{FFT})$.
First, we solve \textsc{FFT} with binary search over $v$ to find the lowest bid $v^{*}$ for which there still exists a feasible team. 

Next, we need to find the team bidding  $v^{*}$ that is preferred by the tie-breaking rule.
We recall that $\coal \prec \coal'$ if $\coal$ precedes $\coal'$ in the lexicographic order. We consider the agents in the increasing order of their names. For each agent $i$ we decrease her salary by 1 ($\minsalary := \minsalary - 1$) and solve \textsc{FFT} for $v = v^{*}-1$. If there is one, this means that in the initial setting there exists a feasible team offering bid $v^{*}$ and having agent $i$ as a member. We store $i$ as a member of the winning team. With the modified salary of $i$ and an updated budget of $v^{*} = v^{*}-1$ we consider the next agent. Otherwise we reset the agent salary $\minsalary$ and the budget $v^{*}$ to their previous values and consider the next agent.

For the second part of the proposition, note that solving the problem of finding the winning team requires solving \textsc{FCFT} with the minimal salaries of the agents set to their asking salaries ($\iminsalary{i} = \salary_i$).
\end{proof}

\theoremstyle{plain}
\newtheorem*{proptwostub0}{Proposition~\ref{proptwo}}

\begin{proptwostub0}
In the project salary model, if the salaries are rational, then finding a Strong Nash Equilibrium can be solved in time $O(n^3\log(nv)\textit{FCFT}))$, where $\textit{FCFT}$ is the complexity of the problem \textsc{FCFT}. Checking whether a given vector of the asking salaries $\langle \salary_i \rangle, i \in N$ is a Strong Nash Equilibrium can be solved in time $O(\textit{FCFT})$, where $\textit{FCFT}$ is the complexity of the problem \textsc{FCFT}.
\end{proptwostub0}
\begin{proof}
Let us start from analyzing the complexity of finding a Strong Nash Equilibrium.
First, we solve a single instance of the \textsc{FCFT} problem to find $N^{*} = N_{\coal^{*}(N)}$. Next, as in the proof of Theorem~\ref{thm:existenceOfSNE}, we introduce the variables $\salary_i, i \in N^{*}$ and inequalities (\ref{inq:salary})-(\ref{inq:cheap2}). If we find the values $\salary_i, i \in N^{*}$ satisfying all the inequalities, then the values $\salary_i, i \in N^{*}$, together with the minimal salaries of the agents outside of $N^{*}$, will form a Strong Nash Equilibrium.

The set of inequalities given in the proof of Theorem~\ref{thm:existenceOfSNE} is a linear program; there are, however, exponentially many constraints (a constraint for each possible team). We construct a separation oracle by a polynomial reduction to 
\textsc{FCFT}. Since the ellipsoid method requires  $O(n^{3}L)$ calls to the separation oracle (where $L$ is the size of the representation of the problem; here $L = O(\log(nv))$), this allows us to solve the linear program in time $O(n^{3}\log(nv) FCFT)$. 

To check whether all the inequalities are satisfied, it is sufficient to solve \textsc{FCFT} with the following parameters. The minimal salaries of the agents from $N^{*}$ are set to $\salary_i$ ($\forall_{i \in N^{*}} \iminsalary{i} := \salary_i$). The minimal salaries of the agents outside of $N^{*}$ are left unmodified. Let $\coal$ denote the solution of such instance of the \textsc{FCFT} problem. There exists a not-satisfied inequality if and only if $N_\coal \neq N^{*}$. The not-satisfied inequality is the inequality that corresponds to the team $\coal \neq \coal^*$. This completes the proof.

Now, let us analyze the complexity of the problem of checking whether a given vector of the asking salaries $\langle \salary_i \rangle, i \in N$ is a Strong Nash Equilibrium.
First, we find a winning team $\coal$ for $\langle \salary_i \rangle$. According to Proposition~\ref{prop:centralWinningTeamRed} we can do this by solving an instance of the \textsc{FCFT} problem (with $\forall i:  \iminsalary{i} := \salary_i$) . Next, we solve another instance $I_2$ of the \textsc{FCFT} problem with the parameters set as follows. We set minimal salaries of the agents from $N_\coal$ to their asking salaries ($\forall_{i \in N_\coal} \iminsalary{i} := \salary_i$). The minimal salaries of the agents outside of $N_\coal$ are left unmodified. If the solution to $I_2$ consists of the members of $N_\coal$ only, we claim that a vector $\langle \salary_i \rangle, i \in N$ is a Strong Nash Equilibrium. Otherwise, it is not.
\end{proof}

\theoremstyle{plain}
\newtheorem*{theoremCharStub}{Theorem~\ref{thm:rigorWinningCoal}}

\begin{theoremCharStub}
The team $\coal$ is rigorously strongly winning if and only if $\coal$ is not explicitly nor implicitly endangered by any team.
\end{theoremCharStub}
\begin{proof}
$\Longleftarrow$ Assume that there exists a rigorously strongly winning team $\coal$; thus  there exists a Rigorously Strong Nash Equilibrium $\mathit{RSNE}$ in which the agents from $N_{\coal}$ get positive payoffs. This implies that the agents from $N_{\coal}$ agree on the action $\langle N_{\coal}, \salary_{\coal}, \bid_\coal \rangle$; other agents ($N \setminus N_{\coal}$) have zero payoffs. For the sake of contradiction let us assume that there exists a feasible team $\coal'$ such that $\coal$ is explicitly or implicitly endangered by $\coal'$.

If $N_{\coal} \cap N_{\coal'}$ is empty ($\coal$ is explicitly endangered by $\coal'$), then $N_{\coal'}$ must be cheaper. This however contradicts the assumption that the agents from $N_{\coal}$ get positive payoffs.

Assume thus that $N_{\coal} \cap N_{\coal'}$ is non-empty (i.e., $\coal$ is implicitly endangered by $\coal'$).
Consider the following collaborative action of agents $(N \setminus N_{\coal}) \cup N_{\coal'}$. All  the agents from $N_{\coal'}$ make action $\coal'$. Each agent $i$ from $N \setminus (N_{\coal} \cup N_{\coal'})$ makes an action $\langle \{\}, \salary_{\emptyset} \rangle$, where $\salary_{\emptyset}$ is an empty function. We show that after playing this action no  agent from $(N \setminus N_{\coal}) \cup N_{\coal'}$ will get lower payoff and that some agents will get a strictly better payoff (which will contradict the assumption that $\mathit{RSNE}$ is a Rigorously Strong Nash Equilibrium). Clearly each agent from $N \setminus (N_{\coal} \cup N_{\coal'})$ does not decrease her payoff (as previously it was equal to 0). Now, we show that the agents from $N_{\coal'}$ will get at least the same payoff as before. Since we know that $\coal$ is implicitly endangered by $\coal'$ (and thus the agents from $N_{\coal} \cap N_{\coal'}$ get in $\coal'$ at least as good payoff as in $\coal$) it is sufficient to show 
that the agents from $N_{\coal'}$ will get positive payoffs.
Indeed, there is no feasible team that includes some agents from $N \setminus (N_{\coal} \cup N_{\coal'})$ (as these agents play $\{\}$). Also, the agents from $N_{\coal} \setminus N_{\coal'}$ do not agree on the collaborative action (they still play $\coal$) and thus, cannot form a feasible team. Thus, after such change of played actions $\coal'$ is the only feasible team that the members agreed on. Finally, we can show that at least one agent will get a strictly better payoff. Either $N_{\coal} = N_{\coal'}$ (and since $\salary_\coal \neq \salary_{\coal'}$, some agent must get a different payoff) or $N_{\coal} \neq N_{\coal'}$ (and the agents from $N_{\coal'} \setminus N_{\coal}$ will get a positive payoff).

$\Longrightarrow$ Assume that $\coal$ is not explicitly nor implicitly endangered by any team. First, if the agents from $N_{\coal}$ make the collaborative action $\coal$, then they will all get positive payoffs. Indeed, the agents in $N_{\coal}$ could not get positive payoffs only if there would exist a cheaper feasible team $\coal'$ such that $N_{\coal} \cap N_{\coal'} = \emptyset$. This would, however mean that $\coal$ is explicitly endangered by $\coal'$. Next, we show that the state in which the agents from $N_{\coal}$ make the collective decision $\coal$ and the other agents play arbitrary actions is RSNE. For the sake of contradiction let us assume that there exists a subset of agents $N_{\coal'}$ which can make a collaborative action $\coal'$ after which the payoff of everyone from $N_{\coal'}$ would be at least equal to her payoff in $\coal$. This would, however mean that $\coal$ is either implicitly or explicitly endangered by $\coal'$. This completes the proof.
\end{proof}

\theoremstyle{plain}
\newtheorem*{lemma1stub}{Lemma~\ref{thm:rigorStrongWinningCharacter}}

\begin{lemma1stub}
In the project salary model, the set of agents participating in a rigorously strongly winning team is the same as the set of agents participating in the cheapest feasible team.
\end{lemma1stub}
\begin{proof}
Let $\coal$ denote the cheapest feasible team. We show that for any other team $\coal'$, such that $N_{\coal} \neq N_{\coal'}$, $\coal'$ cannot be rigorously strongly winning.
For the sake of contradiction let us assume that $\coal'$ is rigorously strongly winning. Let $N_{\cap} = N_{\coal} \cap N_{\coal'}$. Since $\coal$ is the cheapest, the sum of salaries 
of the agents from $N_{\coal} \setminus N_{\cap}$ in $\coal$ is lower or equal to the sum of salaries  of the agents from $N_{\coal'} \setminus N_{\cap}$ in $\coal'$. 
Consider a team $\coal''$ consisting of the set of agents $N_{\coal}$ and the following salary function. The salary of each agent from $N_{\coal} \setminus N_{\coal'}$ is the same as in $\coal$ and the salary of each agent from $N_{\cap}$ is the same as in $\coal'$.
Since the bid $\cost_{\coal'}$ of $\coal'$ was below $v$, the bid of $\coal''$ is also below $v$. Thus, $\coal''$ is feasible. Also, $\coal'$ is implicitly endangered by $\coal''$, which leads to contradiction and completes the proof.
\end{proof}

\theoremstyle{plain}
\newtheorem*{lemma2stub}{Lemma~\ref{prop:maxBid}}

\begin{lemma2stub}
In the project salary model
the bid of a strongly winning team is equal to the maximal allowed price $v$.
\end{lemma2stub}
\begin{proof}
Let $\coal$ be a strongly winning team. If $b_{\coal} < v$ we could increase the salaries of some participating agents. The resulting team would implicitly endanger $\coal$.
\end{proof}

\theoremstyle{plain}
\newtheorem*{prop5stub}{Proposition~\ref{prop:prop5stub}}

\begin{prop5stub}
There exists a weakly winning team if and only if there exists a feasible team.
\end{prop5stub}
\begin{proof}
Consider a feasible team $\coal$ that is not explicitly endangered (such a team exists provided there exists a feasible team). Let $\mathcal{E}$ denote a set of feasible teams implicitly endangering $\coal$. If $\mathcal{E} = \emptyset$, $\coal$ is strongly winning and, thus also, weakly winning. If there exists $\coal' \in \mathcal{E}$ such that $\coal'$ is not (implicitly or explicitly) endangered by any feasible team, then $\coal'$ is strongly winning (and, thus also, weakly winning). Otherwise, $\coal$ is weakly winning.

If there is no feasible team then there is no weakly winning team.
\end{proof}

\theoremstyle{plain}
\newtheorem*{prop6stub}{Proposition~\ref{prop:prop6stub}}

\begin{prop6stub}
In the project salary model, if the salaries of the agents can be rational numbers, the problem of finding a weakly winning team and the problem of checking whether a team $\coal'$ is weakly winning can be solved in time $O(n^5\log(nv)\textit{FCFT})$.
\end{prop6stub}
\begin{proof}
Consider the problem of finding a weakly winning team.
First, we look for a rigorously strongly winning team. If there is one, it is also weakly winning, and so the procedure is complete. If there is no rigorously strongly winning team it is sufficient to find a team that is not explicitly endangered by any other team. We can do this by solving a single instance of the \textsc{FCFT} problem.

Next, consider the problem of checking whether a team $\coal'$ is weakly winning.
We first check whether the team is explicitly endangered by any other team. We can do this by solving a single instance of the \textsc{FCFT} problem for the set of agents $N \setminus N_{\coal'}$.

Now, we look for a rigorously strongly winning team that endangers $\coal'$. We do this in the same way as in the proof of Proposition~\ref{prop:findingRigStrWinn}. The only difference is that we additionally introduce the following inequalities. We assume the same notation as in the proof of Proposition~\ref{prop:findingRigStrWinn}. For each $i \in N_\coal \cap N_{\coal'}$ we require:
$\salary_{\coal}(i) \geq \salary_{\coal'}(i)$.
\end{proof}

\theoremstyle{plain}
\newtheorem*{thm8stub}{Theorem~\ref{thm:newValue}}

\begin{thm8stub}
If there exists a feasible team, then there exists a budget $v^*$ for which there exists a strongly winning team.
The problem of finding such $v^*$ can be solved in time $O(\log v \cdot \textit{FFT})$, where $\textit{FFT}$ is the complexity of the problem \textsc{FFT}.
\end{thm8stub}
\begin{proof}
Let $v^*$ be the smallest value such that there exists a feasible team. We show that for $v^*$ there exists a strongly winning team. Let $\coal^*$ be the most preferred (according to the tie-breaking rule $\prec$) feasible team for $v^*$. For the sake of contradiction let us assume that there exists a team $\coal'$ such that $\coal^*$ is strictly implicitly or explicitly endangered by $\coal'$. Of course $b_{\coal'} \leq v^*$ (otherwise $\coal'$ would not be feasible). If $\coal^*$ is explicitly endangered by $\coal'$ ($N_{\coal^*} \cap N_{\coal'} = \emptyset$), it means $\coal'$ is cheaper than $\coal^*$; and we get a contradiction with the definition of $v^*$. Otherwise ($\coal^*$ is strictly implicitly endangered by $\coal'$), let $i \in N_{\coal^*} \cap N_{\coal'}$. Now, $i$ must get strictly better salary in $\coal'$ than in $\coal^*$. Thus if we change the salary of $i$ in the team $\coal'$ to $\salary_{\coal'}(i) =  \salary_{\coal^*}(i)$ we get a contradiction---a cheaper 
feasible team.

To find such a $v^*$, one has to run a binary search over $v$.
\end{proof}

\theoremstyle{plain}
\newtheorem*{prop9stub}{Proposition~\ref{prop:findAuctionWinning}}

\begin{prop9stub}
The problem of checking whether a feasible team $\coal$ is auction-winning can be solved in time $O(\textit{FFT})$. The problem of finding an auction-winning team can be solved in time $O(v \cdot \textit{FFT})$; $\textit{FFT}$ is the complexity of the problem \textsc{FFT}.
\end{prop9stub}
\begin{proof}
To check whether a team $\coal$ is auction-winning one has to solve the problem of existence of the feasible team for the asking price: $v = b_{\coal} - 1$ (representing the next asking price in the first-price auction); and for each $i \in N_{\coal}$ set $\iminsalary{i} = \salary_{\coal}(i)$ (these agents must get at least the same payoffs as in $\coal$). If no such team exists, $\coal$ is auction-winning.

To find an auction-winning team one can simply simulate the auction. 
\end{proof}

\begin{proposition}\label{prop:sneNotExistsInHourly}
In the hourly salary model there may not exist a Strong Nash Equilibrium even though there exists a feasible team.
\end{proposition}
\begin{proof}
Let us consider the following instance. The budget is $v = 49$. There are 3 agents: $a$, $b$, and $c$; their minimal hourly salaries are $\iminsalary{a} = \iminsalary{b} = \iminsalary{c} = 1$.
All two-agent teams can complete the project: if $a$ and $b$ cooperate they can complete the project spending on it $t_a = 10$ and $t_b = 10$ time units, respectively; if $a$ and $c$ cooperate they must spend $t_a = 22$ and $t_c = 2$ time units; if $b$ and $c$ cooperate they must spend $t_b = 2$ and $t_c = 38$ time units.

For the sake of contradiction let us assume that there exists a Strong Nash Equilibrium. First, consider the case when the agents $a$ and $b$ get positive payoffs in SNE. By the budget constraint, $\salary_b \leq 3$. If $\salary_b = 3$, then  $\salary_a = 1$. The total cost of $\{a, b\}$ is 40. However, $c$, by playing $\salary_c = 1$ can form a cheaper team $\{a, c\}$ with the total cost 24. If $\salary_b \leq 2$ and $\salary_a = 1$, then $a$ has an incentive to play higher. If $\salary_b \leq 2$ and $\salary_a \geq 2$, then $b$ and $c$ are better off by playing a collaborative action with $\salary_b=3$ and $\salary_c=1$---after such an action a team $\{b, c\}$ is cheaper ($\cost_{b,c}=44$) than $\{a,b\}$ ($\cost_{a,b} \geq 50$) and $\{a,c\}$ ($\cost_{a,c} \geq 46$). Thus, $a$ and $b$ cannot both have positive payoffs in SNE.

Second, assume that the agents $a$ and $c$ get positive payoffs in SNE. The total cost of $\{a, c\}$ is $22\salary_a + 2\salary_c$. In such case, if $b$ plays $\salary_a$ then the new team $\{a, b\}$ with total cost $10\salary_a + 10\salary_a$ forms a new cheapest team. 

Finally consider the case when $b$ and $c$ get positive payoffs in SNE. This means that $\salary_c = 1$. But $a$, by playing $1$ can form a team $\{a, c\}$ with the total cost 24. This completes the proof.
\end{proof}

\begin{proposition}\label{prop:checkingRigStrWinn}
Checking whether a team is rigorously strongly winning can be solved in time $O(n^2 \cdot \textit{FCFT})$, where $\textit{FCFT}$ is the complexity of the problem \textsc{FCFT}.
\end{proposition}
\begin{proof}
Let us assume that we want to check whether the team $\coal$ is rigorously strongly winning. First, we check whether we can increase the salary of any agent so that the team would still be feasible. If we can, $\coal$ is not rigorously strongly winning.
Otherwise, we solve \textsc{FCFT} for the set of agents $N \setminus N_{\coal}$. If there exists a non-empty solution $\coal'$ with the cost $\cost_{\coal'} < \cost_\coal$ or such that $\cost_{\coal'} = \cost_\coal$ and $\coal' \prec \coal$, this means that $\coal$ is explicitly endangered by $\coal'$, and thus is not rigorously strongly winning. Otherwise, $\coal$ is not explicitly endangered by any team. 

Next, we check whether $\coal$ is implicitly endangered by some team $\coal'$. We change the names of the agents so that the agents from $N_{\coal}$ were the first $\|N_{\coal}\|$ agents in the lexicographic order. Now, for each agent $i$ from $N_{\coal}$ we do the following procedure. We solve \textsc{FCFT} for the set of agents $N \setminus \{i\}$, for the minimal salaries of the agents from $N_{\coal}$ changed to their salaries in $\coal$, and for the budget $v$ set to $\cost_\coal$.
If there exists a feasible $\coal'$ to \textsc{FCFT} such that the set $N_{\coal'}$ overlaps with $N_{\coal}$ (overlapping can be tested in time $O(n)$), then $\coal$ is implicitly endangered by $\coal'$. 
We already know that there is no non-overlapping team with the cost lower than $\cost_\coal$.
Thus, if for no agent $i$ from $N_{\coal}$ we find such implicitly endangering team, this means that there is no feasible team $\coal'$ such that $N_{\coal} \cap N_{\coal'} \neq \emptyset$. Thus, in such case we conclude that $\coal$ is rigorously strongly winning.
\end{proof}

\begin{proposition}\label{prop:findingRigStrWinn}
In the project salary model, if the salaries of the agents can be rational numbers, finding a rigorously strongly winning team can be solved in time $O(n^5\log(nv)\textit{FCFT})$, where $\textit{FCFT}$ is the complexity of the problem \textsc{FCFT}.
\end{proposition}
\begin{proof}
First we solve \textsc{FCFT} to find the cheapest team $\coal$. We know that the set of the agents participating in a rigorously strongly winning team is $N_\coal$ (Lemma~\ref{thm:rigorStrongWinningCharacter}) and the total cost of such a team is $v$ (Lemma~\ref{prop:maxBid}). We only need to find the salary function of such a team. For every agent $i$ from $N_\coal$, we introduce a variable $\salary_{\coal}(i)$. We will show the linear program for the variables $\salary_{\coal}(i)$, to which the solution is a rigorously strongly winning team. At the same time we will show how to implement the separation oracle for the linear program. 

First equality states that the salaries of the agents satisfy the feasibility constraint:
\begin{align}
\sum_{i \in N_\coal}\salary_{\coal}(i) = v
\end{align}
Next two inequalities model explicit endangerment. For each team $\coal'$, such that $N_\coal \cap N_{\coal'} = \emptyset$ and $\coal' \prec \coal$:
\begin{align}
\sum_{i \in N_\coal}\salary_{\coal}(i) < \sum_{i \in N_{\coal'}}\iminsalary{i} \textrm{.}
\end{align}
For each team $\coal'$, such that $N_\coal \cap N_{\coal'} = \emptyset$ and $\coal \prec \coal'$:
\begin{align}
\sum_{i \in N_\coal}\salary_{\coal}(i) \leq \sum_{i \in N_{\coal'}}\iminsalary{i} \textrm{.}
\end{align}
Note that we can check the above two inequalities by solving \textsc{FCFT} problem for the set of agents $N \setminus  N_\coal$. If the resulting team $\coal'$ is cheaper than $\coal$, this means that the inequality constraint for $\coal'$ was violated. Otherwise, all the above inequalities are satisfied.

Last, for each team $\coal'$, such that $N_\coal \cap N_{\coal'} \neq \emptyset$ and $N_\coal \neq N_{\coal'}$ we introduce the inequality modeling implicit endangerment:
\begin{align}
\sum_{i \in N_{\coal} \setminus N_{\coal'}}\salary_{\coal}(i) + \sum_{i \in N_{\coal'} \setminus N_{\coal}}\iminsalary{i} > v \textrm{.}
\end{align}
We can check this inequality in the same way as we checked whether the team was implicitly endangered in the proof of Proposition~\ref{prop:checkingRigStrWinn}: by swapping the names of the agents, for each $i \in N_\coal$ solving \textsc{FCFT} for the set of agents $N \setminus \{i\}$, and checking the overlapping of the appropriate sets. The whole procedure requires the time $O(n^2 \cdot \textit{FCFT})$.

As the result, we showed the reduction of the problem of finding a rigorously strongly winning team to the linear program with $n$ variables and a separation oracle running in time $O(n^2 \cdot \textit{FCFT})$.
\end{proof}

\section{Other Solution Concepts}\label{sec:otherSolConcepts}

In this section we give a brief overview of other solution concepts that can be applied to describe winning teams in our games.
Most of these solution concepts have their drawbacks and they do not allow to determine winning teams. On the other hand, we point out two ideas that, we believe, are interesting for further study.
The first idea is to apply the concept of the Coalitional Farsighted Conservative Stable Set to our setting. 
The second is to apply the concepts inspired by the graph interpretations. These two solution concepts are, however, more involved, and, so, we believe that our definition of a weakly winning coalition is the natural simplification, and the first step to understand the complexity of the agents' interactions.

In the following subsections we present the discussion on the application of different solution concepts to our model.

\subsection{Cooperative Game Theory Approach}\label{sec:cooperativeGames}

It may seem that our solution concepts are closely related to solution concepts from the cooperative game theory. For instance, the definition of Rigorously Strong Nash Equilibrium is close in spirit to the concept of the core from the cooperative games. However, there are some substantial differences. In cooperative game theory it is commonly assumed that the value of a coalition (in the cooperative game theory teams correspond to coalitions) depends only on the members of this coalition. The following example shows that this is not the case in our problem.

\begin{example}
Consider 2 agents $a$ and $b$ with the minimal salaries $\iminsalary{a} = 1$ and $\iminsalary{b} = 2$. The maximal budget of the issuer is $v = 2$. Consider two team formed by single agents $\coal_1 = \{a\}$, and $\coal_2 = \{b\}$. Let us assume that $\coal_2$ is feasible. The value of $\coal_2$ depends on whether the agent $\coal_1$ is feasible or not.
\end{example}

The above example encourages one to consider our problem as a cooperative game with externalities. However, in such games the values of the coalitions depend only on the partition of the agents into coalitions. In our case, however, the whole coalitions are strategic, and their values depend on the actions (the bids) of the other coalitions. We provide a detailed discussion regarding applicability of selected concepts from cooperative game theory in the two following subsections.

\subsection{The Core}

Although the notion of \emph{the core} is initially known from the cooperative game theory, there is a natural generalization to strategic games.
In this generalization we say that team $\coal$ with payoff function $\payoff$ is in the core if and only if there is no feasible team $\coal'$ with payoff function $\payoff'$ such that every agent in $\coal'$ gets, according to $\payoff'$, a better payoff than according to $\payoff$.

Although, in cooperative game theory we use a simplified model in which feasibility means just that the total payoff of the agents does not exceed the value of the team (i.e., the bid of the team, in our approach), we may use the more demanding notion of feasibility from our model. As a result, a team $\coal$ is in the core if and only if it is not implicitly endangered by any other team.

Intuitively, the notion of the core in our games is missing an important element. Indeed, a team $\coal$ might be in the core even though some other team $\coal'$, disjoint with $\coal$, can offer a better price and, consequently, win the auction and be awarded the project.

\subsection{The (Farsighted) Stability}

Another notion known from the cooperative game theory that is worth considering is the von Neumann-Morgenstern stable set.  The stable set is the set of all payoff vectors such that (i) no payoff vector in the stable set is dominated by another vector in the set, and (ii) all payoff vectors outside the set are dominated by at least one vector in the set.

In the light of our previous example from Proposition~\ref{prop:noWinningCoal}, it is even more appealing to consider the farsighted von Neumann-Morgenstern stable set. A farsighted coalition is more deliberative, it considers that if it makes a deviation, the second team might react as a consequence of the first team's action, next the third team might react, and so on without the limit. In the original formulation the agents are considered to be optimistic---they are willing to deviate if the deviation starts some sequence of deviations that would lead to a better outcome.

In our games the vN-M stable set, and the farsighted vN-M stable set, might be empty.

\begin{example}
Consider the example from Proposition~\ref{prop:noWinningCoal}. There is a project with the budget $v=5$; and three identical agents $a$, $b$, $c$ with minimal salaries $\minsalary > 2$. Every team formed by any two agents is feasible. For the sake of clarity of the presentation let us assume that the payoffs of the agents can be the natural numbers only.
Let us consider the team $\coal_1 = \{a, b\}$ with the payoffs $\salary^{a} = 3$, and $\salary^{a} = 2$. If the team $\coal_1$ is in the stable set, then the team $\coal_2 \{b, c\}$ with the payoffs $\salary^{b} = 3$, and $\salary^{c} = 3$, which dominates $\coal_2$, must not be in the stable set (otherwise it would contradict the internal stability requirement). Since $\coal_2$ does not belong to the stable set, and it is dominated only by the team $\coal_3 = \{a, c\}$ with the payoffs $\salary^{a} = 2$, and $\salary^{c} = 3$, we infer that $\coal_3$ must belong to stable set. However, $\coal_3$ is dominated by $\coal_1$, which leads to contradiction. By symmetry, we see that the stable set is empty.
\end{example}

The same reasoning as given in the example above applies to the farsighted vN-M stable sets. The alternative definition in which the agents are conservative corresponds to the Coalitional Farsighted Conservative Stable Set. Intuitively, in this definition the agents are willing to deviate only if every sequence starting from this deviation leads to a better outcome for them.

We believe that these two cases consider too extreme behavior of the agents. Nevertheless, we think that considering coalitional farsighted conservative stable sets in our game is a very appealing direction for the future work.

\subsection{Coalition-Proof Nash Equilibria}

Another way of weakening the notion of the (rigorously) strongly winning team is to consider Coalition-Proof Nash Equilibria.
Intuitively, in the Coalitional-Proof Nash Equilibrium we first assume that all players are in a common room, where they can freely discuss their strategies. Then the agents, one by one, leave the room. Once an agent leaves the room, she cannot change her strategy. The agents that are left in the room are allowed to discuss and (cooperatively) change their strategies.

Unfortunately, these equilibria are not guaranteed to exist.
This is what we expect since a Coalition-Proof Nash Equilibrium must be essentially a Nash Equilibrium. For the sake of completeness of the presentation, below we show an appropriate example in which there is no Coalition-Proof Nash Equilibrium.

\begin{example}
Consider the example from Proposition~\ref{prop:noWinningCoal}. There is a project with the budget $v=5$; and three identical agents $a$, $b$, $c$ with minimal salaries $\minsalary > 2$. Every team formed by any two agents is feasible.
There is no Coalition-Proof Nash Equilibrium in this example (independently whether the salaries of the agents are natural or rational numbers). Indeed, consider any vector of payoffs $\langle \salary^{a}, \salary^{b}, \salary^{c} \rangle$. If $\salary^{a} > 2$, we infer that $a$ forms a winning team with one of the agents $b$, or $c$. Without loss of generality we assume that $\{a, b\}$ is the winning team. Thus, $\salary^{b} < 3$ and $\salary^{c} = 0$. If we consider the subgame formed by the agents $b$ and $c$, we see, however, that their payoff vector $\langle \salary^{b}, \salary^{c} \rangle$ is Pareto-dominated by $\langle 3, 2 \rangle$. Now, let us consider the case when $\salary^{a} < 2$. One of the agents $b$ and $c$ needs to have payoff lower than $3$ (w.l.o.g let us assume that this is the agent $b$). But, if we consider the subgame formed by the agents $a$ and $b$, their payoff vector $\langle \salary^{a}, \salary^{b} \rangle$ is Pareto-dominated by $\langle 2, 3 \rangle$.
Finally, let us assume that $\salary^{a} = 2$. We infer that one of the agents $b$ and $c$ gets zero payoff (let us assume that this is the agent $b$). However, the payoff vector $\langle \salary^{a}, \salary^{b} \rangle$ is Pareto-dominated by $\langle 3, 2 \rangle$.
\end{example}

\subsection{Graph Interpretations}

Let us consider a directed multi-graph in which the vertices are the strategy profiles. Each pair of vertices can be connected with at most two edges, corresponding to implicit and explicit endangerment. Thus, vertices $v$ and $u$ are connected by an edge corresponding to the implicit endangerment if and only if $v$ is implicitly endangered by $u$. Analogously, $v$ and $u$ are connected by an edge corresponding to the explicit endangerment if and only if $v$ is explicitly endangered by $u$.

Clearly, in such a graph, strong Nash equilibria correspond to the sinks, the vertices with no outgoing edges. Also, the edges corresponding to explicit endangerment do not form cycles. Consequently, we can restrict our graphs to these induced by the vertices that do not have outgoing edges corresponding to the explicit endangerment. We believe that every connected component in such restricted graphs defines an interesting set of stable solutions. We plan to analyze this idea in our future work. For instance, thus defined set of stable solutions is always non-empty and its elements correspond to weakly winning teams.

\section{Finding Feasible Teams in a Scheduling Model}\label{sec:findingFeasible}

In Sections~\ref{sec:winningCoalitions}~and~\ref{sec:mechDesign} we show that many problems of finding the (weakly/strongly) winning teams or determining whether a given team is (weakly/strongly)
winning require solving the subproblem of finding the feasible team.
The general model (Section~\ref{sec:model}) assumed that given a team there is an oracle deciding whether there exists a feasible team.

By specifying an oracle, our results can be applied to two different problems known in the literature: the commodity auctions and the path auctions.

In the commodity auctions setting, the project can be seen as a set of items $I = \{i_1, i_2, \dots, i_q\}$ , where each agent owns a certain subset of the items. A team is feasible if the agents have together all the items from $I$.

In the path auctions setting~\cite{Nisan:1999:AMD:301250.301287} we are given a graph $G$ with two distinguished vertices: a source $s$ and a target $t$. The agents correspond to the vertices in the graph. A team is feasible if the participating agents form a path from $s$ to $t$.

In this section we show a possible concrete instance of this model in which a project is a set of indivisible, independent, tasks and agents are processors who process these tasks with varying speeds.

\subsection{The Scheduling Model}
A project consists of a set  $\mathcal{T} = \{t_1, t_2, \dots, t_q\}$ of $q$ independent tasks. The tasks can be processed sequentially or in parallel. The tasks are indivisible: a task must be processed on a single processor. Once started, a task cannot be interrupted.
All tasks must be completed before a given time $d$, the project's deadline.  

Agents correspond to processors (in this section we use terms ``agent'' and ``processor'' interchangeably). Each agent has certain skills which are represented as the speed of executing the tasks. Thus, for each agent $i$ we define the skill vector $s_i = \langle s_{i, 1}, s_{i, 2}, \dots s_{i, q} \rangle$ which has the following meaning: agent $i$ is able to finish task $t_j$ within $s_{i, j}$ time units (with $s_{i, j}=\infty$ when the agent is unable to finish the task).
We assume that $s_i$ is known for each agent (it can be well approximated, e.g., from past behavior of the agents certified by clients in form of reviews).
An agent can process only a single task at each time moment---if she wants to process more than one task, she must execute the tasks sequentially. 
We assume that only a single agent can work on a given task. This assumption is not as restrictive as it may appear; if the task $t_i$ is large and can be processed by multiple agents in parallel, the project client will rather replace $t_i$ by a number of smaller tasks.

For a team $\coal$ we define $\Phi_\coal: \mathcal{T} \rightarrow N_{\coal}$ to be an assignment function (assigning tasks to agents). The assignment function $\Phi_\coal$ enables us to formalize the notion of a team completing the project before the deadline and also the total cost of the team.
Specifically, a project is finished before the deadline $d$ if and only if all the agents finish their assigned tasks before $d$, $\forall i \in N_{\coal}: \sum_{\ell: \Phi(t_{\ell}) = i} s_{i, \ell} \leq d$.
In the hourly salary model, the cost of the team is equal to $\cost_{\coal} = \sum_{i \in N_{\coal}} \salary_{\coal}(i) \sum_{\ell: \Phi(t_{\ell}) = i} s_{i, \ell}$. 

In the scheduling model we define the problem of finding a feasible team as follows.

\begin{problem}[\textsc{FFTSM}: Find Feasible Teams, Scheduling Model]\label{prob:findingGen}
Let $\mathcal{T}$ be the set of $q$ tasks and $N$ be the set of processors (or equivalently, agents). For each task $t_j \in \mathcal{T}$ and each processor (agent) $i \in N$ we define $s_{i,j}$ as the processing time of $t_j$ on $i$. Let $\iminsalary{i}$ be the cost of renting processor $i$ (hiring agent $i$). The budget of the project is $v$ and the deadline is $d$.
The FFTSM problem consists of selecting a subset of the processors $N' \subseteq N$ and the assignment function $\Phi: \mathcal{T} \rightarrow N'$ such that the budget is not exceeded ($\cost_{N', \Phi} \leq B$) and the project's makespan does not exceed the deadline $d$.
\end{problem}

In the hourly salary model, the problem of finding the feasible team reduces to the problem of scheduling on unrelated processors with costs. Specifically, there exists a 2-approximation algorithm for approximating the makespan (the deadline $d$ in our model). 

\begin{problem}[\textsc{FFTHS}: Find Feasible Teams, Hourly Salary]\label{prob:finding3}
The instances of the problem are the same as in the \textsc{FFTSM} problem, except that in the \textsc{FFTHS} problem we additionally specify that the cost of the team $\cost_{N', \Phi}$ is defined as $\cost_{\coal} = \sum_{i \in N_{\coal}} \salary_{\coal}(i) \sum_{\ell: \Phi(t_{\ell}) = i} s_{i, \ell}$.
\end{problem}
 
The project salary model is a generalization of the problem of minimizing makespan on unrelated processors. To the best of our knowledge, this problem has not been stated before; thus we formally define it below.

\begin{problem}[\textsc{FFTPS}: Find Feasible Teams, Project Salary]\label{prob:finding2}
The instances of the problem are the same as in the \textsc{FFTSM} problem, except that in the \textsc{FFTPS} problem we additionally specify that the cost of the team $\cost_{N', \Phi}$ is defined as $\cost_{N', \Phi} = \sum_{i \in N'}\iminsalary{i}$.
\end{problem}

An easier variant of the problem, in which the goal is to optimize the assignment only (assuming that the processors are already selected) has a 2-approximation algorithm. However, adding the notion of the budget usually significantly increases the complexity. We believe that the approximability of \textsc{FFTPS} is a very appealing problem.

\subsection{FFTPS: Hardness Results}

First, we show the NP-hardness of \textsc{FFT}-Scheduling  in restricted special cases.

\begin{theorem}\label{thm:hardnessNP1}
\textsc{FFTPS} and \textsc{FFTHS} are NP-hard even for two agents.
\end{theorem}
\begin{proof}
The proof is by reduction from the partition problem. In the partition problem, we are given a set of integers $\{ n_j \}$; we ask whether there exists a partition of this set into two subsets $S_1, S_2$, such that $\sum_{n_j \in S1} n_j = \sum_{n_j \in S_2} n_j$. To construct an instance of the feasible team problem, we construct a project that has a task for each $n_j$, an unlimited budget and a deadline $d= 1/2 \sum n_j$. We take two agents $a$ and $b$ with processing speeds $s_{a,j} = s_{b,j} = n_j$ and unit costs: $\iminsalary{a} = \iminsalary{b} = 1$. A feasible team corresponds with partitioning numbers into two with equal sums.
\end{proof}

\begin{theorem}\label{thm:hardnessNP2}
\textsc{FFTPS} is NP-hard even if the agents can be assigned no more than 3 tasks, if each agent has no more than 3 skills (for each $j$ we have that $\|\{i: s_{i,j} \neq \infty \}\| \leq 3$), if the deadline is constant, and if the minimal salaries of the agents are equal 1.
\end{theorem}
\begin{proof}
The proof is by reduction from the exact set cover problem. In the exact set cover problem we are given a set of elements $T = \{t_1, t_2, \dots, t_q\}$ and family $\mathcal{S} = \{S_1, S_2, \dots, S_n\}$ of $3$-element subsets of $T$. We ask whether there exist $\frac{q}{3}$ subsets from $\mathcal{S}$ that cover all the elements from $T$. The exact set cover problem is NP-hard even if each member of $T$ appears in at most 3 sets from $\mathcal{S}$.

We build an instance of the feasible team problem in the following way. There are $q$ tasks and $n$ agents; for each agent $i$ and each task $t_j$ we have that $s_{i, j} = 1$ if and only if $t_j \in S_i$. Otherwise, $s_{i, j} = \infty$. The deadline $d$ is equal to 3. The minimal salary of each agent is 1 and the budget $v$ to $\frac{q}{3}$. It is easy to check that there exists a feasible team if and only if there exists a cover of $T$ with $\frac{q}{3}$ sets.
\end{proof}

\begin{theorem}\label{thm:hardnessNP3}
\textsc{FFTHS} is NP-hard even if the agents can be assigned no more than 4 tasks, if each agent has no more than 4 skills (for each $j$ we have that $\|\{i: s_{i,j} \neq \infty \}\| \leq 4$), if the deadline is constant, and if the minimal salaries of the agents are equal 1.
\end{theorem}
\begin{proof}
The proof is by reduction from the exact set cover problem. We are given a set of elements $T = \{t_1, t_2, \dots, t_q\}$ and family $\mathcal{S} = \{S_1, S_2, \dots, S_n\}$ of $3$-element subsets of $T$. We assume that each member of $T$ appears in at most 3 sets from $\mathcal{S}$.

We build an instance $I$ of the feasible team problem in the following way. There are $q + n$ tasks and $2n$ agents. The first $q$ tasks $t_1, t_2, \dots, t_q$ correspond to the elements in $T$. The next $n$ tasks $t_{q+1}, t_{q+2}, \dots t_{q+n}$ are the dummy tasks needed by our construction. The first $n$ agents $1, 2, \dots, n$ correspond to the subsets from $\mathcal{S}$ and the next $n$ agents $(n+1), (n+2), \dots, 2n$ are the dummy agents. The minimal salaries of all agents are equal to 1. 

For each agent $i$, $i \leq n$ and each task $t_j$, $j \leq q$, we set $s_{i, j} = 2$ if and only if $t_j \in S_i$; otherwise $s_{i, j} = \infty$. Also, for each agent $i$, $i \leq n$ and each task $t_j$, $j > q$ we set $s_{i, j} = 5$ if and only if $i = j-q$; otherwise $s_{i, j} = \infty$. For each agent $i$, $i > n$ and each task $t_j$ we set $s_{i, j} = 6$ if and only if $i-n = j-q$; otherwise $s_{i, j} = \infty$. The deadline $d$ is equal to $6$ and the budget $v$ is equal to $v = \frac{7}{3}q+5n$. Clearly, each agent has no more than 4 skills and so, in any feasible solution, cannot be assigned more than 4 tasks.

We will show that the answer to the original instance of the exact set cover problem is ``yes'' if and only if there exists a feasible team in the our constructed instance $I$.

$\Longleftarrow$ Let us assume there exists a feasible team $\coal$. The cost of this team is at most equal to $v = \frac{7}{3}q+5n$. Each non-dummy task (there are $q$ such tasks) takes 2 time units, and thus implies the cost equal to $2$. The dummy tasks can be assigned either to non-dummy agents (implying the cost 5) or to dummy agents (implying the cost 6). Thus, we infer that at most $\frac{q}{3}$ dummy agents are assigned a task ($2q + \frac{1}{3}q\cdot 6 + (n - \frac{1}{3}q) \cdot 5 = v$). As the result at least $(n - \frac{q}{3})$ dummy tasks must be assigned to non-dummy agents. A non-dummy agent, who is assigned a dummy task cannot be assigned any other task (otherwise the completion time would exceed the deadline). Thus, at most $\frac{q}{3}$ non-dummy agents can be assigned non-dummy tasks. The non-dummy tasks can be assigned only to non-dummy agents. We see the subsets corresponding to these non-dummy agents who are assigned non-dummy tasks form the solution to the initial exact set 
cover problem.

$\Longrightarrow$ Let us assume that there exists the exact set cover in the initial problem. The agents corresponding to the subsets from the cover can be assigned tasks so that the deadline is not exceeded and the total cost of completing these tasks is equal to $2q$. The other $(n - \frac{q}{3})$ non-dummy agents can be assigned one dummy task each. Finally, not-yet assigned dummy tasks can be assigned to dummy agents. The total cost of such assignment is equal to  $2q + (n - \frac{1}{3}q) \cdot 5 + \frac{1}{3}q\cdot 6 = v$.

This completes the proof.
\end{proof}

Unfortunately, \textsc{FFTPS} is not approximable for makespan, for budget, and even for the combination of both these parameters. 

\begin{theorem}\label{thm:nonApprox}
For any $\alpha, \beta \geq 1$ there is no polynomial $\alpha$-$\beta$-approximation algorithm for \textsc{FFTPS} that approximates makespan with the ratio $\alpha$ and budget with the ratio $\beta$, unless P=NP. This result holds even if the costs of all processors are equal 1.
\end{theorem}
\begin{proof}
For the sake of contradiction let us assume that there exists $\alpha$-$\beta$-approximation algorithm $A$. 
We provide a reduction showing that $A$ can be used as $\beta$-approximation algorithm for \textsc{SetCover}, a contradiction with well-known lower bound of $\mathrm{lm}(n)$ on approximating \textsc{SetCover}. Let $I$ be an instance of \textsc{SetCover}, where $T = \{t_1, t_2, \dots, t_q\}$ is the set of elements and $\mathcal{S} = \{S_1, S_2, \dots, S_n\}$ is the set of the subsets of $T$. We ask whether there exists $K$ subsets from $\mathcal{S}$ that together cover all elements from $T$.

From $I$ we construct an instance of \textsc{FFTPS} in the following way. There are $q$ tasks corresponding to $q$ elements in $I$. There are $n$ agents $1, 2, \dots, n$ corresponding to the subsets in $\mathcal{S}$. The duration $s_{i,j}$ of the task $t_i$ when processed by the agent $j$ is defined in the following way. If $t_i \in S_j$ then $s_{i,j} = 1$. Otherwise, $s_{i,j} = \alpha q + 1$. The minimal salary of each agent is equal to 1 and the total budget is $K$.
We show that if there exist $K$ subsets from $\mathcal{S}$ covering $T$ then we can use $A$ to find $\beta K$ subsets covering $T$.

Let $C$ denote the covering using $K$ subsets. If we assign each task $t_i$ to any agent $j$ such that $S_j \in C$ and $t_i \in S_j$, then the completion time of the tasks on each processor will be at most equal to $q$. In such case we will use only $K$ processors. Thus $A$ returns the solution with the makespan at most equal to $\alpha q$ using at most $\beta K$ processors. This, however, means that each task $t_i$ is assigned to such agent $j$ that $t_i \in S_j$. Thus, the subsets corresponding to the selected processors form the solution of $I$. Of course, there is at most $\beta K$ such processors. This completes the proof.
\end{proof}

Theorems~\ref{thm:hardnessNP1},~\ref{thm:hardnessNP2},~and~\ref{thm:hardnessNP3} show that the problems \textsc{FFTPS} and \textsc{FFTHS} remain NP-hard even if various parameters are constant. Although Theorem~\ref{thm:hardnessNP1} gives us NP-hardness even for 2 agents, it is somehow not satisfactory as we used the fact that the deadline $d$ can be very large. If the deadline is given in unary encoding, we can solve the case for 2 agents by dynamic programming. Thus, it is interesting if we can solve the problem efficiently for small numbers of agents, if the input is given in unary encoding. We use parameterized complexity theory to approach this problem. We ask if \textsc{FFTPS} and \textsc{FFTHS} have $\fpt$ algorithms for the parameter $n$, the number of the agents, provided the input is given in unary encoding. 

\begin{theorem}
Consider the number of agents as the parameter. \textsc{FFTPS} and \textsc{FFTHS} are W[1]-hard, even if all the agents have minimal salaries equal to 1, and if the size of the input is given in unary encoding.
\end{theorem}
\begin{proof}
We show the reduction from Unary Bin Packing (which is $\wone$-hard). In the instance of the unary bin packing problem we are given a set $T$ of $q$ items $T = \{t_1, t_2, \dots, t_q\}$ (the size of the item $t_i$ is equal to $s_i$) and a set $N$ of $n$ bins, each having a capacity $d$. We ask whether it is possible to pack all the items to the bins.

From this instance we can construct the instance of \textsc{FFTPS} (or \textsc{FFTHS}) in the following way. Here $T$ will be the set of tasks, $N$ will be the set of agents. The minimal salaries of the agents are equal to 1; the speed of processing the task $t_j$ by the agent $i$ is equal to $s_{i, j} = s_j$ . In \textsc{FFTPS} we set the total budget $v$ to be equal to $n$. In \textsc{FFTHS} we set $v$ to $\sum_{t_i \in T}s_i$. Of course, there exists a feasible schedule if and only if there exists a feasible bin-packing.  
\end{proof}

\subsection{Integer Programming Formulation}
In this subsection we state the FFTPS problem as an integer programming problem for the hourly salary model.

\begin{alignat}{2}
\text{minimize }   & d\ \label{eq:ip_1}\\
\text{subject to } & \sum_{i \in N} a_i \iminsalary{i} \leq v\ \label{eq:ip_2}\\
                   & x_{i, j} \leq a_i\ &,\ & i \in N \label{eq:ip_3}\\
                   & \sum_{t_j \in T} x_{i, j} s_{i, j} \leq d\ &,\ & i \in N \leq d \label{eq:ip_4}\\
                   & x_{i, j} \in \{0, 1\}\ &,\ & i \in N; t_j \in T \label{eq:ip_5}\\
                   & a_i \in \{0, 1\}\ &,\ & i \in N \label{eq:ip_6}
\end{alignat}

In the above formulation, a binary variable $a_i$ denotes whether agent $i$ is a part of the solution (is assigned some tasks, Equation~\ref{eq:ip_6}). A binary variable $x_{i, j}$ is equal to 1 if and only if the task $t_j$ is assigned to the agent $i$ (Equation~\ref{eq:ip_5}). We minimize the makespan $d$ (Equation~\ref{eq:ip_1}), which is the maximal completion time of the tasks over all the agents (Inequality~\ref{eq:ip_4}). We cannot exceed the budget $v$ (Inequality~\ref{eq:ip_2}), and the tasks can be assigned only to the selected agents (Inequality~\ref{eq:ip_3}).

\end{document}